\documentclass[12pt]{article}
\usepackage[utf8]{inputenc}
\usepackage{graphicx,array,amsmath,tikz,amssymb,amsthm}
\usepackage[margin=0.75in ]{geometry}
\usepackage{ifthen}
\usepackage{tikz-3dplot}
\usepackage{caption,subcaption}
\usepackage{authblk}

\newtheorem{thm}{Theorem}[section]
\newtheorem{prop}[thm]{Proposition}

\newtheorem{conj}[thm]{Conjecture}

\newtheorem{cor}[thm]{Corollary}

\DeclareMathOperator{\dist}{dist}

\begin{document}

\title{Classification of Open and Closed Convex Codes on Five Neurons}
\author[1]{Sarah Ayman Goldrup\thanks{sarah96.sg@gmail.com}} 
\author[1]{Kaitlyn Phillipson\thanks{kphillip@stedwards.edu}}
\affil[1]{Department of Mathematics, St.\ Edward's University}

\maketitle

\begin{abstract}
Neural codes, represented as collections of binary strings, encode neural activity and show relationships among stimuli. Certain neurons, called place cells, have been shown experimentally to fire in convex regions in space. A natural question to ask is: Which neural codes can arise as intersection patterns of convex sets? While past research has established several criteria, complete conditions for convexity are not yet known for codes with more than four neurons. We classify all neural codes with five neurons as convex/non-convex codes. Furthermore, we  investigate which of these codes can be represented by open versus closed convex sets. Interestingly, we find a code which is an open but not closed convex code and demonstrate a minimal example for this phenomenon.
\end{abstract}

\section{Introduction}\label{sec:intro}
Understanding neural firings patterns and studying what they represent is an important problem in neuroscience. Neural activity can be modeled via \emph{neural codes}, which are binary patterns representing the recorded neural activity.  Neural codes show  relationships between stimuli, like distance between locations in an environment. Through neural codes, the brain is able to characterize and map the physical world.

 In 1971, O'Keefe discovered place cells in the hippocampus, which is a part of the brain that processes and stores memories and is involved in navigation. \emph{Place cells} are a special type of neuron that form internal maps of the external world. O'Keefe found that place cells exhibited high firing rate when the rat was in a specific area in space, called the neuron's \emph{place field}.  Through experimental results, it was shown that place fields were approximately convex regions of the space.  O'Keefe was awarded a shared Nobel Prize in Medicine in 2014 for this work (see \cite{placecells} for more detail).
 
Some mathematical questions that arise are: Given a neural code, can it arise from a collection of convex open sets?  Can we find criteria to determine whether a neural code is convex by its combinatorial structure alone?  Furthermore, if a neural code is a convex code, what is the minimal dimension needed to represent the code geometrically? While past results gave necessary and sufficient criteria for a code to be convex, these conditions are incomplete for codes with five or more neurons. In this paper, we completely classify all neural codes on five neurons which are open convex, and give partial results for closed convex codes.
 
 In Section \ref{sec:background}, we review past results on neural codes and provide concise definitions. In Section \ref{sec:catalog}, we give a catalog of all open convex codes on five neurons which were not classified by prior results. In Section \ref{sec:unstable}, we give a new definition for codes which are open convex but not closed convex (called \emph{unstable codes}), and prove that three codes on five neurons are unstable.  We summarize and give open questions in Section \ref{sec:conclusion}.

\section{Background and previous results}\label{sec:background}
We review some notation and definitions pertaining to this problem; see \cite{neuralring} for more detail.

A \emph{codeword} on $n$ neurons is a string of  0's and 1's of length $n$, where 1 denotes neural activity and 0 denotes silence.  We can also write the codeword $\sigma$ as a subset of active neurons $\sigma\subset [n]:=\{1,2,\ldots, n\}$.  Both notations will be used interchangeably.
A \emph{neural code} on $n$ neurons is a collection of codewords $C\subset 2^{[n]}$.  For computational convenience, we will always assume that the ``silent'' codeword $000\cdots 0 (=\emptyset)$ is in $C$.

Given open sets $\mathcal{U} =\{U_1,...,U_n\}\subset \mathbb{R}^d$, the \emph{code of the cover} is the neural code defined as: 
\begin{equation*}
C(\mathcal{U}) = \{\sigma \subseteq [n]: \bigcap_{i\in \sigma} U_i \setminus \bigcup_{j \in [n]\setminus \sigma}  U_j \neq \emptyset\} 
\end{equation*}
Each codeword in $C(\mathcal{U})$ corresponds to intersections of open sets in $\mathcal{U}$ which are not covered by other sets in $\mathcal{U}$. If a neural code $C$ = $C(\mathcal{U})$, and $\mathcal{U}=\{U_1,...,U_n\}$ is a cover with each $U_i$ a convex subset of $\mathbb{R}^d$, then $C$ is a \emph{convex code} and $\mathcal{U}=\{U_1,...,U_n\}$ is a geometric realization of the code $C$.

For example, the code $C_1 =\{\emptyset, 1,2,12\} $ is a convex code because it can be realized as a collection of convex open sets (see Figure \ref {fig:exconvexcode}). The \emph{minimal embedding dimension} is the smallest dimension $d=d(C)$ such that the neural code $C$ is realizable as a convex code in $\mathbb{R}^d$. Note that though $C_1$ is drawn in $\mathbb{R}^2$, $d(C_1)=1$.
 
 \begin{figure}[ht]
\centering
\includegraphics{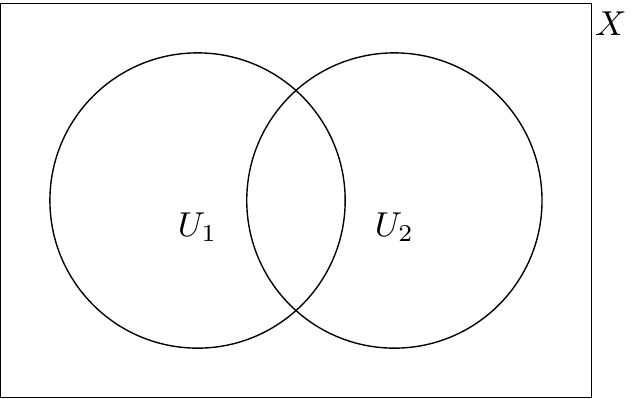}
\caption{An open convex realization of code $C_1$}\label{fig:exconvexcode}
\end{figure}

For a code $C$, we can investigate its intersection structure by constructing its \emph{simplicial complex}:

\begin{equation*}
\Delta(C):=\{\sigma\subseteq [n]:\sigma\subseteq c \;\text{for some}\; c\in C\}.
\end{equation*}

This is the smallest abstract simplicial complex\footnote{A collection $K$ of subsets of a finite set $X$ is an \emph{abstract simplicial complex} if, for every $\alpha \in K$ and $\beta\subset \alpha$, then $\beta \in K$.} which contains all elements of $C$.  Elements of $\Delta(C)$ that are maximal under inclusion are \emph{maximal codewords} (also called \emph{facets}). A code $C$ is \emph{max intersection-complete} if it contains all intersections of maximal codewords in $C$. 

For a face (or codeword) $\sigma \in \Delta$, the \emph{link} of $\sigma$ in $\Delta$ is the simplicial complex: 
\begin{equation*}
Lk_\Delta(\sigma)=\{w \in \Delta : \sigma \cap w= \emptyset, \sigma \cup w \in \Delta\}. 
\end{equation*}

One approach to deciding if a code is not convex is to determine whether a code has an obstruction to convexity due to a topological inconsistency in the intersections of its codewords. This kind of obstruction is called a \emph{local obstruction}, and the process to find these obstructions is given in \cite{curtoetal}. We summarize their findings here.

A simplicial complex is \emph{contractible} if its geometric realization is contractible\footnote{A set is \emph{contractible} if it is homotopy equivalent to a point.}. For a given simplicial complex $\Delta$, we let
\begin{equation*} C_{\min}(\Delta)=\{\sigma\in \Delta: Lk_\Delta(\sigma) \;\text{is non-contractible}\}\cup\{\emptyset\}.\end{equation*} It was shown in \cite{curtoetal} that if $C$ is a any code with simplicial complex $\Delta$, $C$ has no local obstructions if and only if $C_{\min}(\Delta)\subseteq C$. Moreover, they showed that every nonempty element of $C_{\min}(\Delta)$ is an intersection of facets of $\Delta$. Thus, we can consider the non-empty elements of $C_{\min}(\Delta)$ as a collection of the maximal codewords (or facets) of $\Delta$ and non-maximal codewords $\sigma$ such that $Lk_\Delta(\sigma) \;\text{is non-contractible}$. We will call these non-maximal codewords the \emph{mandatory codewords} for $\Delta$, since any convex code with simplicial complex $\Delta$ must contain these codewords.

A complete condition for the convexity of a neural code is still unknown; we summarize here the known results.

\begin{prop}\label{prop:partialcond} For a neural code $C$:
\begin{itemize}
\item[1.] If $C$ is max intersection-complete, then $C$ is convex.
\item[2.]  If $C$ is convex, then $C$ has no local obstructions.

\end{itemize}
\end{prop}

Part 1 of Proposition \ref{prop:partialcond} is due to \cite{openclosed}, while Part 2 is due to \cite{curtoetal}. Note that Part 2 implies that if $C$ is convex, $C_{\min}(\Delta)\subseteq C$. The converses of Part 1 and Part 2 of Proposition \ref{prop:partialcond} hold for $n\leq 4$ (see \cite{curtoetal}); however, these statements fail for $n=5$.  An example of a convex code that is not max intersection-complete is the code $C1$ in Table \ref{table:class}. An example of a non-convex code that has no local obstructions was found in \cite{lien}, which is code $C4$ in Table \ref{table:class}.  

Thus, the classification of which max intersection-incomplete codes with no local obstructions are actually convex remains an open problem. As a step in this direction, we investigate all codes with five neurons that are max intersection-incomplete with no local obstructions. This problem has also been investigated independently in \cite{yagerthesis}.

\section{Classification of open convex codes on five neurons}\label{sec:catalog}

Given a simplicial complex $\Delta$, let $C_{\min}:=C_{\min}(\Delta)$ be the \emph{minimal code} for $\Delta$ defined in Section \ref{sec:background}. It was shown in \cite{openclosed} that open convex codes exhibit \emph{monotonicity} in the following sense: if $C$ is an open convex code, and $D$ is a code such that $C\subseteq D\subseteq \Delta(C)$, then $D$ is also open convex. Therefore, we define a simplicial complex $\Delta$ to be \emph{convex minimal} if the corresponding minimal code of $\Delta$,  $C_{\min}(\Delta)$ is open convex. If a simplicial complex $\Delta$ is convex minimal, then all codes $C$ with $\Delta(C)=\Delta$ and no local obstructions are convex.
 
In \cite{lien}, all unique simplicial complexes on 5 vertices were computed, as well as the maximal facets and the corresponding mandatory codewords for each simplicical complex. It was found that of the 157 unique simplicial complexes, for 22 of these codes, the set of mandatory codewords did not contain all possible intersections of facets. Thus, the minimal code of each of these simplicial complexes is a max intersection-incomplete code with no local obstructions, which cannot be classified by Proposition \ref{prop:partialcond}.

In Table \ref{table:class}, we classify these 22 max intersection-incomplete codes with no local obstructions. For each simplicial complex $\Delta$, we list the maximal codewords, mandatory codewords, and non-mandatory intersections of maximal codewords, and describe the convexity of $C_{\min}(\Delta)$.  We enumerate the list of minimal codes using the designations C$m$, where $m=1,\ldots, 22$.  Columns 2, 3, and 4, were computed in \cite{lien}.    

Of the 22 simplicial complexes with this property, only one minimal code does not have a convex realization: $C4$ was proved to be non-convex in \cite{lien}.  Convex realizations for each code, except C4, are given in \ref{app:geo}. Interestingly, only one code on the list has minimal dimension 3, which is $C22$, the construction of which is given in \ref{app:cleanup}. The remainder of the codes have minimal dimension 1 or 2.

\begin{table}[ht!]{\small
\centering
\begin{tabular}{|c|p{3.10cm}|p{3.10cm}|p{2.44cm}|p{2.72cm}|}
\hline
& Maximal  & Mandatory  & Non-mandatory   &\\
$C_{\min}(\Delta)$ & codewords& codewords& intersections &{\centering Classification} \\
 & (facets) & &of facets &  (see \ref{app:geo}) \\
\hline
\hline
\hline
C1 & 123, 124, 145  & 12, 14 & 1 &  convex minimal  \\ \hline
C2& 134, 135, 234, 12 & 13, 34, 1, 2 & 3 &  convex minimal \\ \hline
C3& 134, 235, 345, 12 & 34, 35, 1, 2 & 3 & convex minimal\\ \hline
C4& 2345, 123, 124, 145  & 12, 14, 23, 24, 45, 2, 4  & 1 &Not open convex (see \cite{lien}) \\ \hline
C5& 123, 124, 145, 234 & 12, 14, 23, 24, 2  & 1, 4 &convex minimal\\ \hline
C6& 123, 125, 145, 234 & 12, 15, 23, 4 & 1, 2 & convex minimal\\ \hline
C7& 145, 234, 245, 12, 13 & 24, 45, 1, 2, 3 & 4 & convex minimal\\ \hline
C8& 135, 145, 235, 12, 34 & 15, 35, 1, 2, 3, 4 & 5 & convex minimal \\ \hline
C9& 134, 135, 145, 234, 12 & 13, 14, 15, 34, 1, 2  & 3, 4 & convex minimal \\ \hline
C10& 134, 135, 234, 245, 12  & 13, 24, 34, 1, 2, 5 & 3, 4 & convex minimal\\ \hline
C11& 134, 135, 245, 345, 12 & 13, 34, 35, 45, 1, 2, 3   & 4, 5 &  convex minimal  \\ \hline
C12& 123, 124, 125, 134, 345   & 12, 13, 14, 34, 1, 5   & 3, 4 & convex minimal \\ \hline
C13& 123, 124, 135, 145, 234 & 12, 13, 14, 15, 23, 24, 1, 2  & 3, 4 & convex minimal \\ \hline
C14& 123, 124, 125, 145, 234 & 12, 14, 15, 23, 24, 1, 2 &4 &convex minimal \\ \hline
C15& 123, 125, 145, 234, 345 & 12, 15, 23, 34, 45   & 1, 2, 3, 4, 5& convex minimal \\ \hline
C16& 145, 235, 345, 12, 13, 24 &  35, 45, 1, 2, 3, 4 &5 & convex minimal \\ \hline
C17& 145, 234, 235, 245, 12, 13 & 23, 24, 25, 45, 1, 2, 3 &4, 5 & convex minimal  \\ \hline
C18& 134, 135, 145, 234, 235, 12 & 13, 14, 15, 23, 34, 35, 1, 2, 3    &4, 5 & convex minimal \\ \hline
C19& 134, 135, 234, 245, 345, 12 & 13, 24, 34, 35, 45, 1, 2, 3, 4   &5 & convex minimal \\ \hline
C20& 123, 124, 125, 135, 145, 234& 12, 13, 14, 15, 23, 24, 1, 2   &3, 4 &convex minimal \\ \hline
C21& 123, 124, 125, 145, 234, 345  &12, 14, 15, 23, 24, 45, 45, 1, 2, 4  &3, 5 & convex minimal \\ \hline
C22& 123, 124, 125, 135, 145, 234, 235 & 12, 13, 14, 15, 23, 24, 25, 35, 1, 2, 3, 5 & 4 & convex minimal with $d(C22)=3$  \\
\hline

\end{tabular}
\caption{Classification of max intersection-incomplete codes with no local obstructions on five neurons}
\label{table:class}}
\end{table}

Combining our findings in Table \ref{table:class} with the property of monotonicity for open convex codes, we immediately obtain the following theorem.

\begin{thm}
Let $C$ be a code on five neurons for which $\Delta(C)$ is not isomorphic to $\Delta(C4)$. Then $C$ is open convex if and only if $C$ has no local obstructions.
\end{thm}

These results can also be summarized by noting the sparsity of the codes. A code $C$ is \emph{$k$-sparse} if $|\sigma|\leq k$ for all $\sigma\in C$. The catalog of codes in Table \ref{table:class} gives the following result:
 
 \begin{thm}
 All 3-sparse codes with no local obstructions on five neurons are open convex.
 \end{thm}
 This gives affirmative evidence to the question introduced in \cite{newlocal}: Is every 3-sparse code with no local obstructions convex? If a counterexample exists, it must have at least 6 neurons.

\section{Closed convex codes and unstable codes}\label{sec:unstable}

\subsection{Background on closed codes}
Section \ref{sec:catalog} gave a complete classification for open convex codes on five neurons. However, the question of classifying \emph{closed convex codes}, which can be realized as a collection of closed convex sets in $\mathbb{R}^d$, is not fully answered. In \cite{openclosed}, it was shown that the results of Proposition \ref{prop:partialcond} can be extended to closed convex codes as well: Max intersection-complete codes are both open and closed convex, and if a code is closed convex, then it has no local obstructions. However, it is unknown whether the monotonicity condition holds for closed convex codes. Note that though $C4$ is not open convex, it has been shown to be closed convex (see \cite{openclosed}).  %Though all the locally good $3$-sparse codes on five neurons are open convex, the same cannot be said for closed convex.

Initially, it was conjectured that every open convex code is also a closed convex code; however,  this is not the case for neural codes with $n > 4$. The authors of \cite{openclosed} showed that the neural code $D =\{\mathbf{123, 126, 156, 456, 345, 234,}\ 12, 16, 56, 45, 34, 23, \emptyset\}$ (with facets given in bold) is an open convex code but not a closed convex code with six neurons.  %An open convex realization is shown in Figure \ref{fig:4}. 

We define \emph{unstable codes} as codes that are open convex but not closed convex. Intuitively, unstable codes arise when the neural code of the open convex realization would change drastically if the open sets were perturbed slightly. It was unknown whether $D$ was the minimal example of an unstable code with respect to the number of neurons. Here, we show that five neurons is the minimal case where unstable codes occur.

%\begin{figure}[ht]
%\centering
%\includegraphics[width= 1.75in]{code6.pdf}
%\caption{An open convex realization of code $D$}\label{fig:4}
%\end{figure}

\subsection{Minimal case of unstable codes}\label{subsec:mincase}
Through inspection, we can see that all the open convex realizations for the codes in Table \ref{table:class} can be viewed as closed convex realizations except for three cases: $C6, C10$, and $C15$. The following result shows that these three codes are indeed not closed convex.

\begin{thm}\label{thm:10}  The following codes are open convex but not closed convex:
\begin{itemize}
\item[] $C6=\{\mathbf{125, 234, 145, 123}, 4, 23, 15, 12, \emptyset \}$ 
\item[] $C10 = \{\mathbf{134, 245, 234, 135, 12}, 1, 5, 34, 13, 2, 24, \emptyset \}$
\item[] $C15 =\{\mathbf{145, 125, 123, 234, 345}, 23, 15, 45, 34, 12, \emptyset\}$
\end{itemize}
\end{thm}
\begin{proof}
Codes $C6, C10$, and $C15$ have open convex realizations given in \ref{app:geo}. We will employ the technique used in \cite{openclosed} to show that $C6$ is not closed convex.

Suppose there is a closed convex cover $U= \{U_i\}_{i=1}^{5}$ in $\mathbb{R}^d$ for $C6$.   For $\sigma\subseteq [5]$, we let $U_\sigma=\cap_{i\in \sigma}U_i$. 
We can pick distinct points $x_{234} \in U_{234}$, and $x_{145}\in U_{145}$. Let $M$ be the line segment connecting $x_{145}$ to $x_{234}$.  Pick $x_{123} \in U_{123}$ so that for every $a \in U_{123}$, we have $\dist(a,M) \geq \dist(x_{123}, M)$, i.e., $\dist(x_{123},M) $ is the minimal distance to $M$. Let $L_1= \overline{x_{123}x_{145}}$. $L_1 \subset U_1$ since $U_1$ is convex, and $U_1\subset U_2\cup U_5$, since whenever neuron 1 appears in a codeword, it will also appear with either 2 or 5. Together, these imply $L_1 \subset U_2 \cup U_5$.  Since $L_1$ is connected and the sets $U_2 \cap L_1$ and $U_5 \cap L_1$ are closed and nonempty, $U_2 \cap U_5 \cap L_1 \subset U_{125}$ is nonempty and there is a point $x_{125} \in U_{125} \cap L_1$ that is on the line segment $L_1$ (see Figure \ref{fig:5}).

\begin{figure}[ht]
\centering
\includegraphics{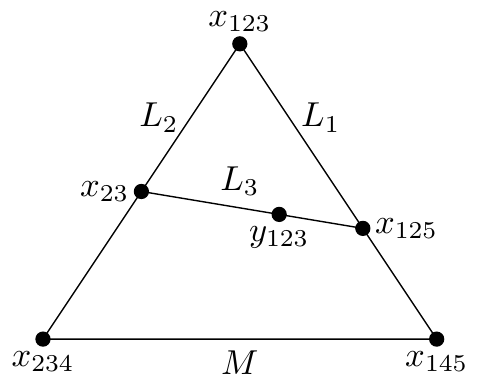}
\caption{Proof of Theorem \ref{thm:10}}\label{fig:5}
\end{figure}

Let $L_2= \overline {x_{123}x_{234}}$. $U_{23}$ is convex so $L_2 \subset U_{23}$. Since $U_{23} \cap L_2 \subset U_{23}$ is nonempty, there is a point $x_{23} \in U_{23} \cap L_2$ that is on line segment $L_2$. Let $L_3= \overline {x_{23}x_{125}}$. By convexity, $L_3 \subset U_2$ and since $U_2\subset U_1\cup U_3$, then $L_3 \subset U_1 \cup U_3$. Since $L_3$ is connected and the sets $U_1 \cap L_3$ and $U_3 \cap L_3$ are closed and nonempty, $U_1 \cap U_{3} \cap L_3 \subset U_{123}$ is nonempty, and there is a point $y_{123} \in U_{123} \cap L_3$ that is on the line segment $L_3$.

We see that the point $y_{123}$ lies in the interior of the closed triangle $\Delta (x_{123},x_{145}, x_{234})$. Therefore, $\dist(y_{123},M) < \dist(x_{123}, M) $ which is a contradiction. This implies that $C6$ cannot be realized as a collection of closed convex sets, and thus it is not a closed convex code. 

The proof for $C10$ is similar, using $x_{245}$ and $x_{135}$ as starting points, and building $x_{234}$ of minimal distance to $\overline{x_{245}x_{135}}$. For $C15$, we can use $x_{145}$ and $x_{345}$ as starting points, and choose $x_{123}$ of minimal distance to $\overline{x_{145}x_{345}}$. The details are left to the reader.
\end{proof}

Since the condition of being max intersection-complete is equivalent to being closed and open convex for $n\leq 4$ neurons (see \cite{openclosed}), we have the following result:

\begin{cor}
If $C$ is an unstable neural code on $n$ neurons, then $n\geq 5.$ This bound is tight.
\end{cor}

\subsection{Discussion on unstable codes}
We end this section by discussing some observations of the previous examples of unstable codes. We note that $D, C6, C10$, and $C15$ are all examples of $3$-sparse codes with at least four maximal codewords. When computing the simplicial complexes of $D, C6, C10, $ and $C15$, we see that they all have at least two distinct non-mandatory intersections of facets. In the proof of Theorem \ref{thm:10}, we used the absence of these codewords in the code in our construction: for a missing non-mandatory intersection of facets $\sigma$, $U_\sigma$ is a set completely contained in the union of other sets in the cover, which was used to construct a new point on each line.  

Note that in the in the open convex realizations of each code $D, C6, C10$, and $C15$, there appears to be lines slicing the plane and meeting in a common point. Moreover, the boundary points of the open convex sets overlap. 

Some differences to note are that $D$ has 6 neurons and 6 maximal codewords, $C6$ has 5 neurons and 4 maximal codewords, and $C10$ and $C15$ both have 5 neurons and 5 maximal codewords. $C15$ and $\Delta(D)$ have more than two non-mandatory intersections of facets. From these observations, and from the similarities in the proof technique used in \cite{openclosed} and \ref{subsec:mincase}, we make the following conjecture:

\begin{conj}\label{conj}
 Let $C$ be a max intersection-incomplete open convex code, where $\Delta(C)$ has at least two non-mandatory intersections of facets not contained in $C$. Suppose $C$ has at least 3 maximal codewords $ M1, M2, M3$, and there is $ \sigma \subset M1$ with $\sigma \in C$ such that $\sigma \cap M2 \not \in C$. Then $C$ is not a closed convex code. \\
%2. If points in $U_ {Mi}$ where $M_i$ are maximal codewords lie on the same line then it is a closed convex code. (This is a code that has at least 1 non-mandatory codeword.)
\end{conj}

\section{Conclusion and future research}\label{sec:conclusion}
In this paper, we showed that all 3-sparse neural codes with no local obstructions on five neurons are open convex. This result also shows that every simplicial complex on five neurons is convex minimal except $\Delta(C4)$.  Furthermore, we showed that unstable neural codes can only occur for binary codes with no fewer than 5 neurons, and thus showed which of the minimal codes are both open and closed convex code and which are unstable codes. Besides Conjecture \ref{conj}, some future problems that can be investigated are:  Does the monotonicity condition hold for closed convex codes (in particular, for the closed convex codes from Table \ref{table:class})? Which codewords need to be added to $C6, C10,$ or $C15$ to make each code closed convex? Are there examples of unstable codes which are not 3-sparse?

\section{Acknowledgments}
Over summer 2017, SAG was supported by the Dr.\ M.\ Jean McKemie Endowed Student/Faculty Fund for Innovative Mathematics Summer Scholarship, and KP was supported by the Presidential Excellence Grant of St.\ Edward's University. We thank Zvi Rosen for his contribution to the realization of $C10$, Luis Garc\'ia Puente for discussion on this topic, and Anne Shiu for helpful comments on previous drafts.

\bibliographystyle{plain}
\bibliography{Bibil}

\appendix

\section{Construction of $C22$}\label{app:cleanup}

\begin{prop}\label{prop:c22}
	$C22=\{145, 124, 135, 235, 125, 123, 234, 12, 13, 14, 15,$ $23, 24, 25, 35, 1, 2, 3, 5, \emptyset \}$ is convex with minimal embedding dimension 3.
\end{prop}
\begin{proof} $C22$ has Helly Dimension 3 (see \cite{curtoetal}) by the hollow simplex formed by 135, 235, 125, 123, which implies that the minimal embedding dimension is at least 3.  This code is indeed realizable as a convex code in three dimensions via the following construction:
	
	Construct a tetrahedron by taking the convex hull of the points $\{(0,0,0), (4,0,0), (0,4,0)$, $(0,0,4)\}$.  Create open sets in $\mathbb{R}^3$ by extending each face into the interior of the tetrahedron with the following thicknesses:
	\begin{itemize}
		\item[] $U_1$ is the face on the $xz$-plane with thickness 1.
		\item[] $U_2$ is the face on the $yz$-plane of thickness 3/4.
		\item[] $U_3$ is the face on the $xy$-plane of thickness 1/2.
		\item[] $U_5$ is the face on the diagonal of thickness 1/4.
	\end{itemize}
	
	This construction is shown in Figure \ref{fig:convexC22}. These sets give all codewords involving $1,2,3,5$ as desired.  Consider the points $p_3=(1,0,3)$ and $p_1=(0,2,0)$.  Under a slight perturbation, $p_3\in U_1\cup U_5\setminus (U_2\cup U_3)$ and $p_1\in U_2\cup U3\setminus (U_1\cup U_5)$.  The goal is to construct the set $U_4$ using the line segment $\overline{p_1p_3}$.
	
	We claim the line segment $L=\overline{p_1p_3}$ is completely contained in $U_1\cup U_2$.  Parametrize the line segment $L=L(t)$ as:
	\begin{align*}
	x&=t\\
	y&=2-2t\\
	z&=3t
	\end{align*}
	Note that $L(0)=p_1$ and $L(1)=p_3$.  It is a straightforward exercise to check that the rest of the line segment is covered as shown in Figure \ref{fig:covering}.

	%%label U_i's
	\begin{figure}[ht]
		\centering
		\includegraphics[width=.99\textwidth]{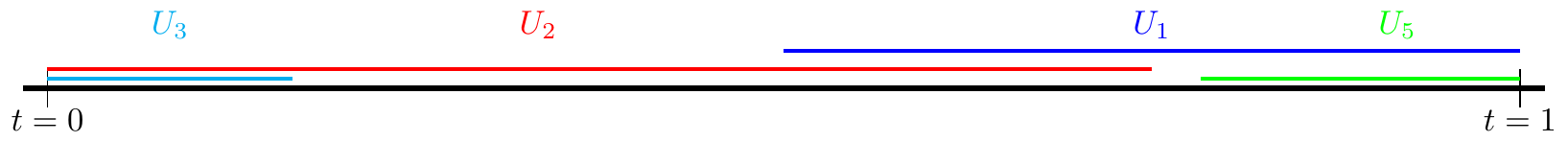}
		\caption{Covering of $L(t)$}\label{fig:covering}
	\end{figure}
	
	For $\epsilon$ sufficiently small, the open cylinder $C(\epsilon, L(t))$ with radius $\epsilon$ and center $L(t)$ will be completely contained in $U_1\cup U_2$.  Thus, if we let $U_4=C(\epsilon, L(t))$, we have our desired realization of the neural code $C22$.
	
	\begin{figure}[ht]
		\centering\includegraphics{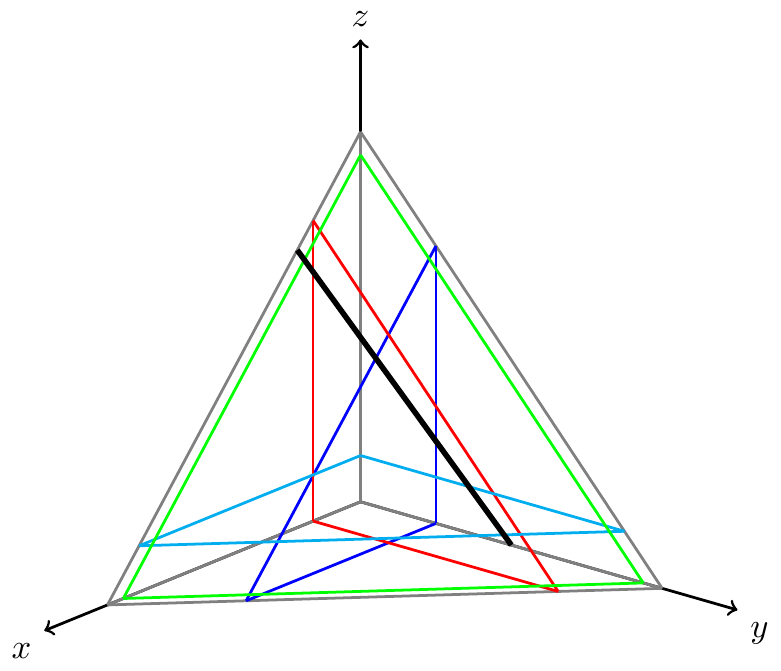}
		\caption{Convex Realization of code C22}\label{fig:convexC22}
	\end{figure}

\end{proof}

\section{Geometric realizations of $n=5$ open convex codes}\label{app:geo}

\begin{figure}[ht]
\centering
\begin{subfigure}{0.45\textwidth}
	\centering\includegraphics[width=.75\textwidth]{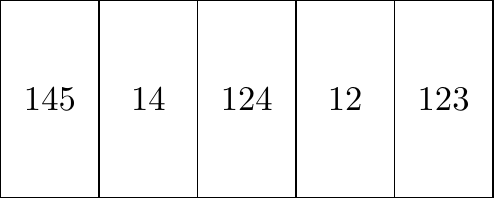}
	\caption*{$C1$}
\end{subfigure}
\begin{subfigure}{0.45\textwidth}
	\centering\includegraphics[width=.75\textwidth]{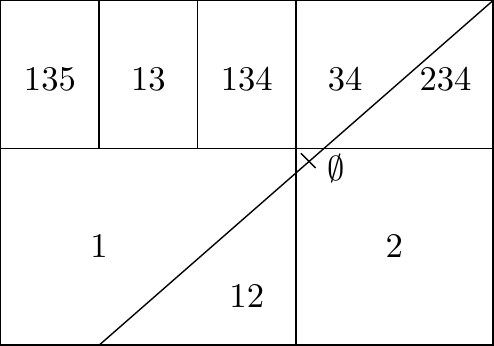}
	\caption*{$C2$}
\end{subfigure}

\vspace{1cm}

\begin{subfigure}{0.45\textwidth}
	\centering\includegraphics[width=.75\textwidth]{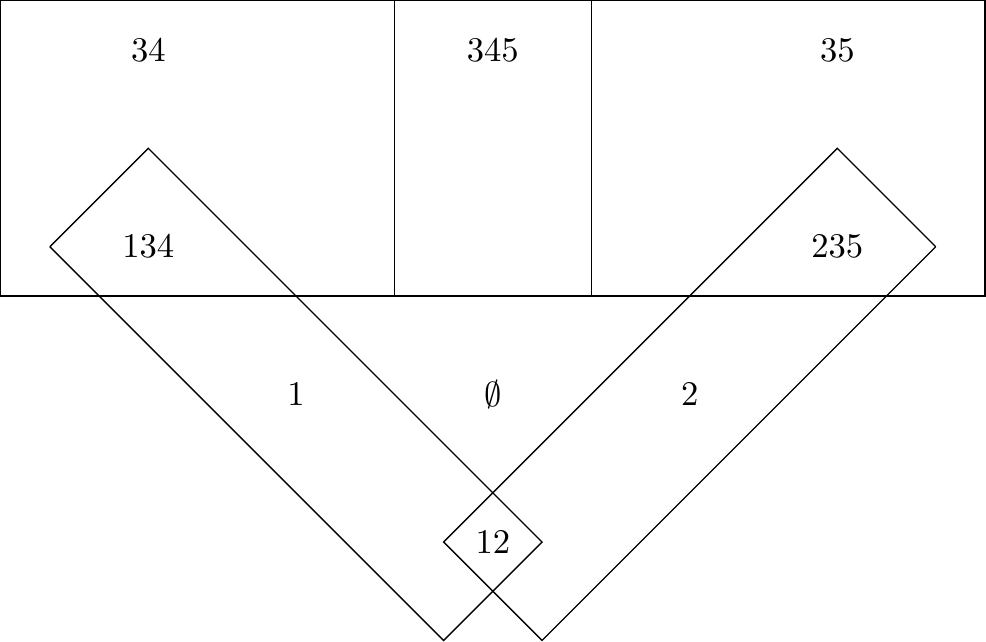}
	\caption*{$C3$}
\end{subfigure}
\begin{subfigure}{0.45\textwidth}
	\centering\includegraphics[width=.75\textwidth]{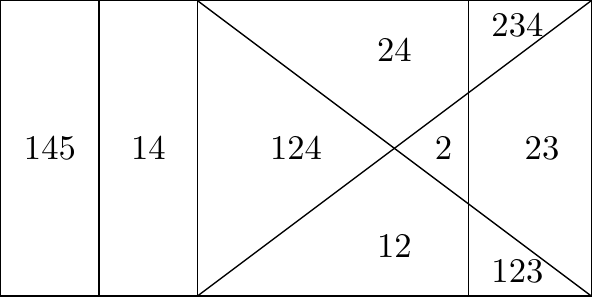}
	\caption*{$C5$}
\end{subfigure}

\vspace{1cm}

\begin{subfigure}{0.45\textwidth}
	\centering\includegraphics[width=.75\textwidth]{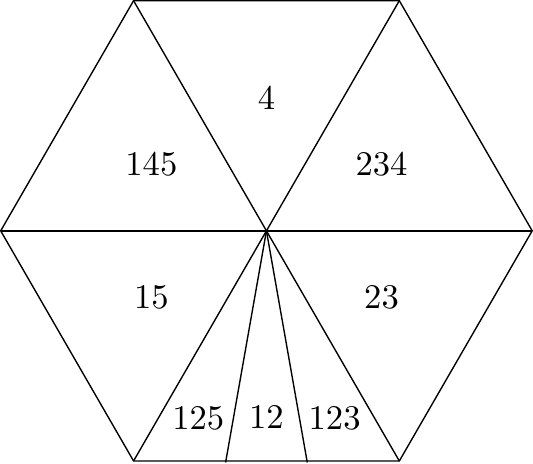}
	\caption*{$C6$}
\end{subfigure}
\begin{subfigure}{0.45\textwidth}
	\centering\includegraphics[width=.75\textwidth]{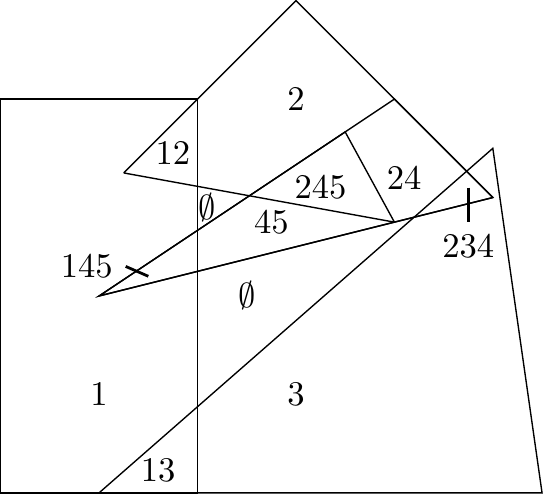}
	\caption*{$C7$}
\end{subfigure}

\end{figure}

\begin{figure}[ht]
\centering

\begin{subfigure}{0.45\textwidth}
	\centering\includegraphics[width=.75\textwidth]{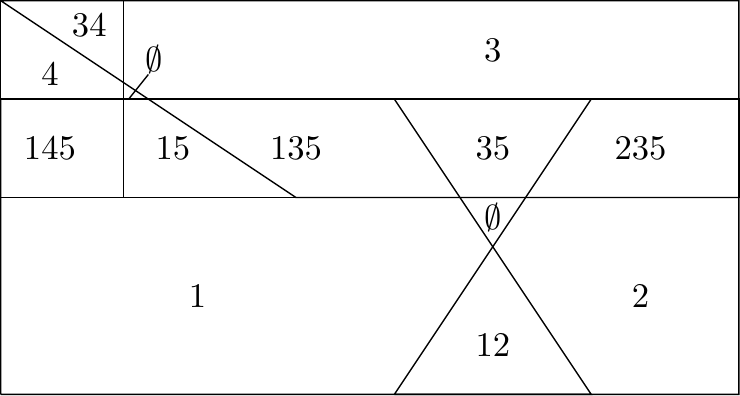}
	\caption*{$C8$}
\end{subfigure}
\begin{subfigure}{0.45\textwidth}
	\centering\includegraphics[width=.75\textwidth]{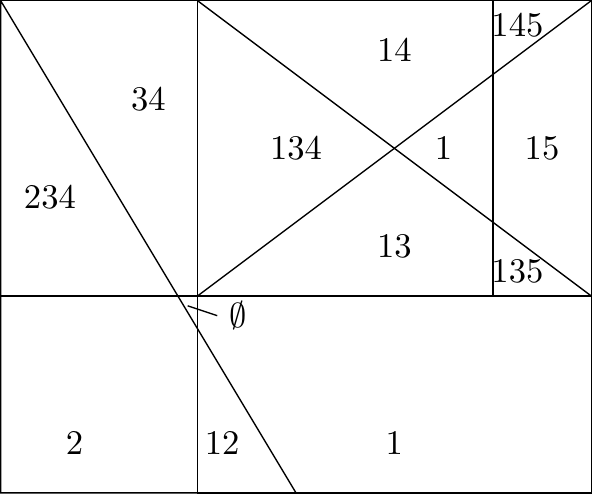}
	\caption*{$C9$}
\end{subfigure}

\vspace{1cm}

\begin{subfigure}{0.45\textwidth}
	\centering\includegraphics[width=.75\textwidth]{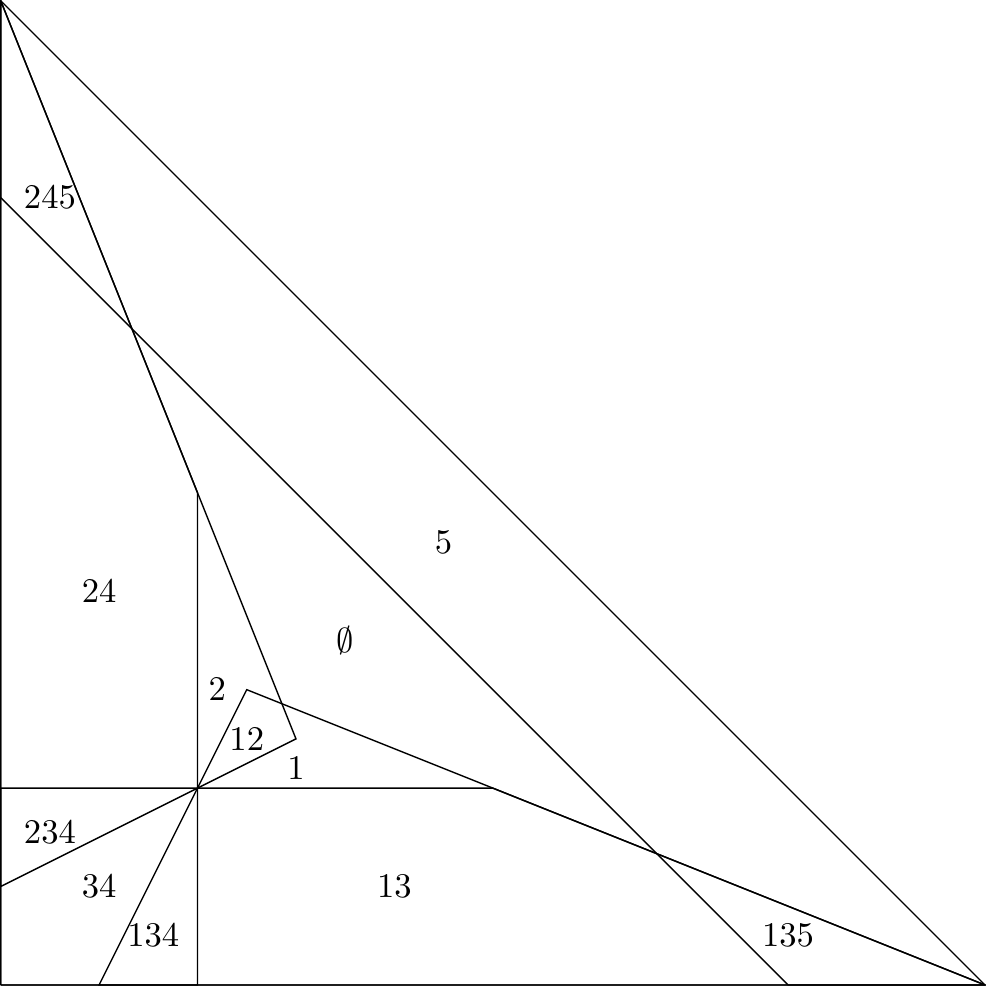}
	\caption*{$C10$}
\end{subfigure}
\begin{subfigure}{0.45\textwidth}
	\centering\includegraphics[width=.75\textwidth]{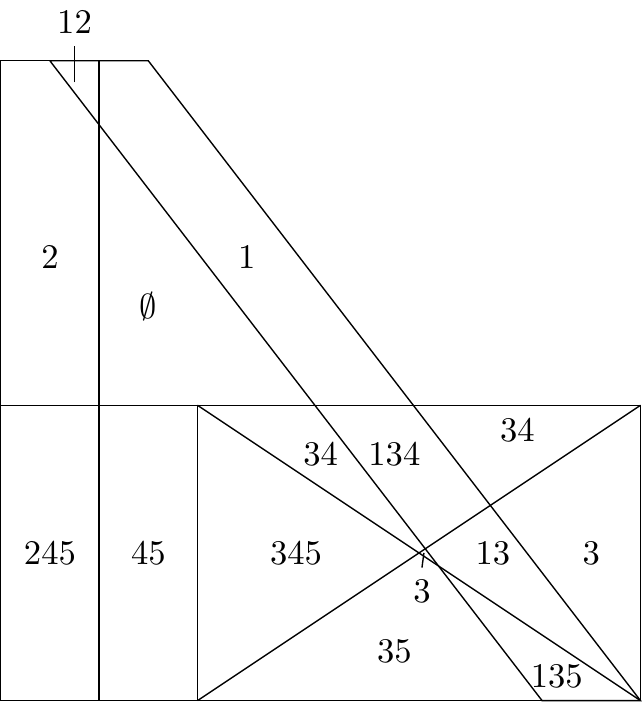}
	\caption*{$C11$}
\end{subfigure}

\vspace{1cm}

\begin{subfigure}{0.45\textwidth}
	\centering\includegraphics[width=.75\textwidth]{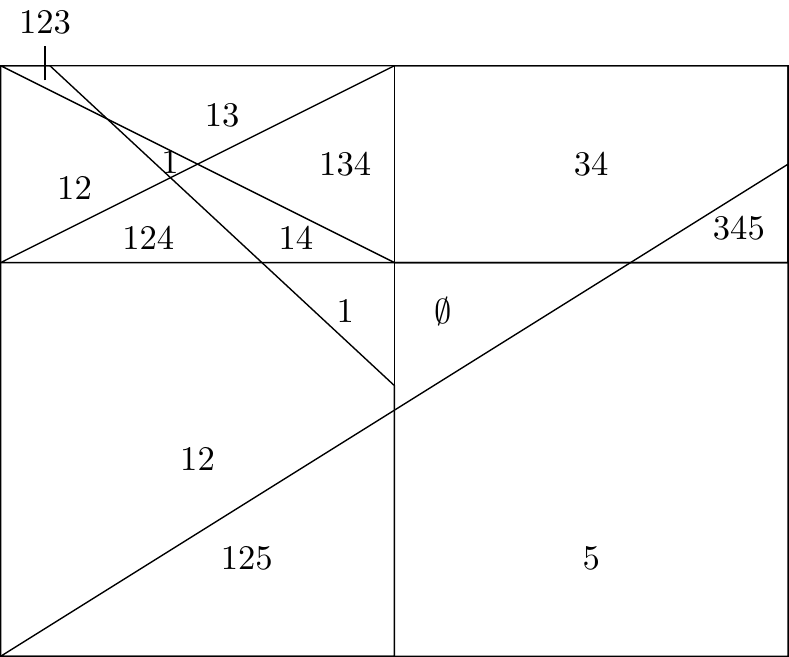}
	\caption*{$C12$}
\end{subfigure}
\begin{subfigure}{0.45\textwidth}
	\centering\includegraphics[width=.75\textwidth]{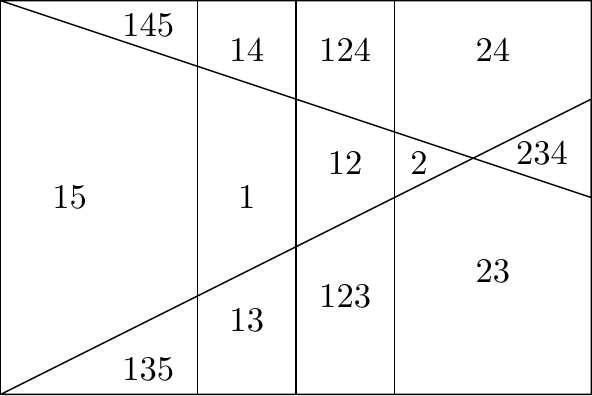}
	\caption*{$C13$}
\end{subfigure}

\end{figure}

\begin{figure}[ht]
\centering

\begin{subfigure}{0.45\textwidth}
	\centering\includegraphics[width=.75\textwidth]{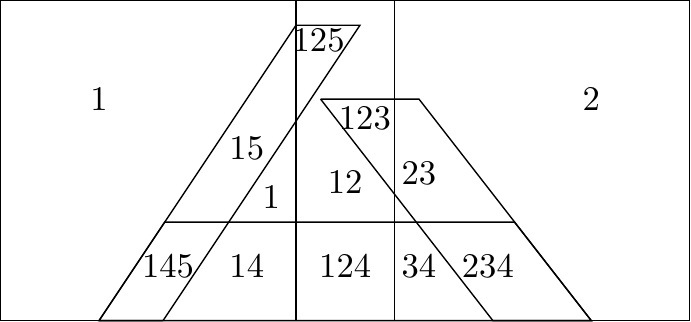}
	\caption*{$C14$}
\end{subfigure}
\begin{subfigure}{0.45\textwidth}
	\centering\includegraphics[width=.75\textwidth]{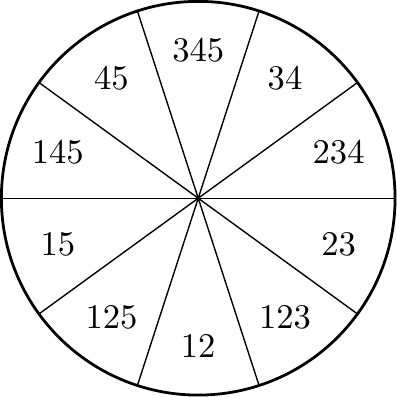}
	\caption*{$C15$}
\end{subfigure}

\begin{subfigure}{0.45\textwidth}
	\centering\includegraphics[width=.75\textwidth]{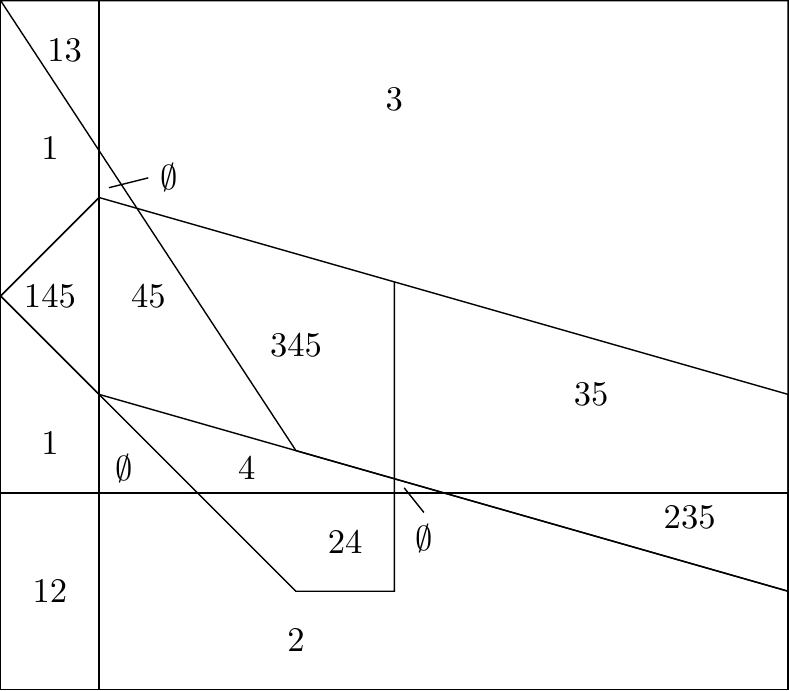}
	\caption*{$C16$}
\end{subfigure}
\begin{subfigure}{0.45\textwidth}
	\centering\includegraphics[width=.75\textwidth]{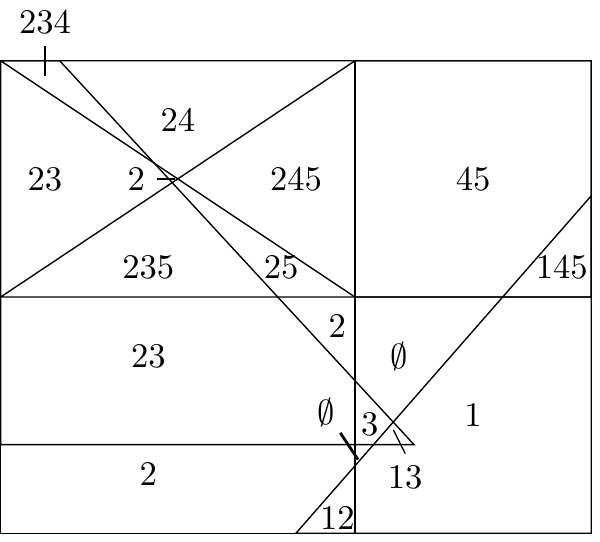}
	\caption*{$C17$}
\end{subfigure}

\vspace{1cm}

	\begin{subfigure}{0.45\textwidth}
		\centering\includegraphics[width=.75\textwidth]{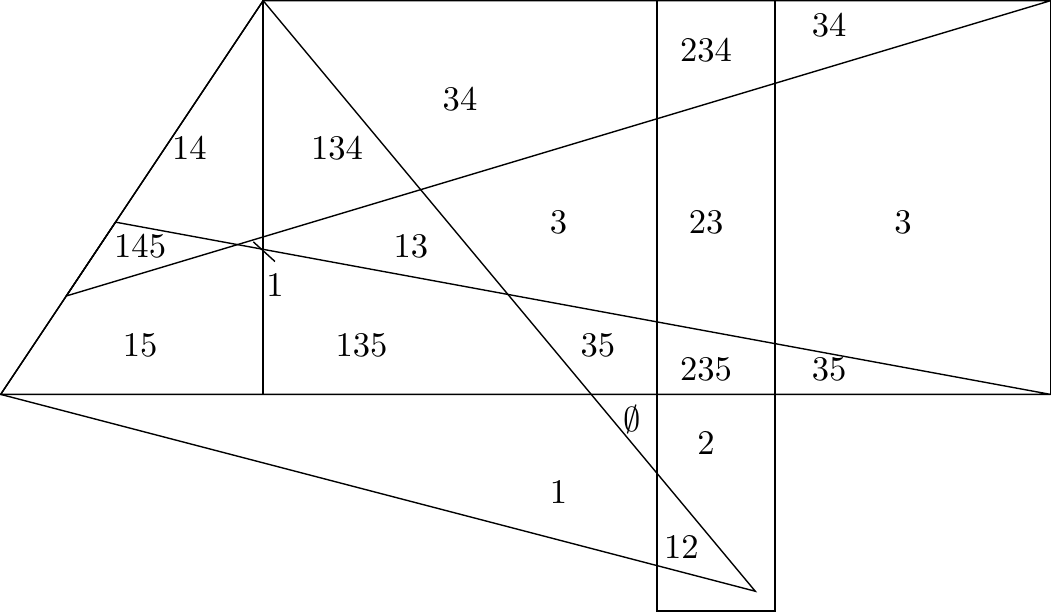}
		\caption*{$C18$}
	\end{subfigure}
	\begin{subfigure}{0.45\textwidth}
		\centering\includegraphics[width=.75\textwidth]{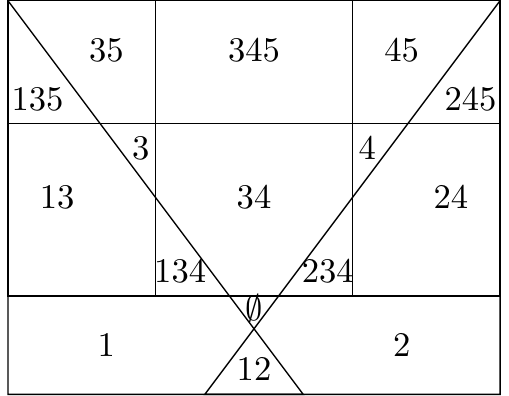}
		\caption*{$C19$}
	\end{subfigure}
\end{figure}

\begin{figure}[ht]
	\centering
	
	\begin{subfigure}{0.45\textwidth}
		\centering\includegraphics[width=.75\textwidth]{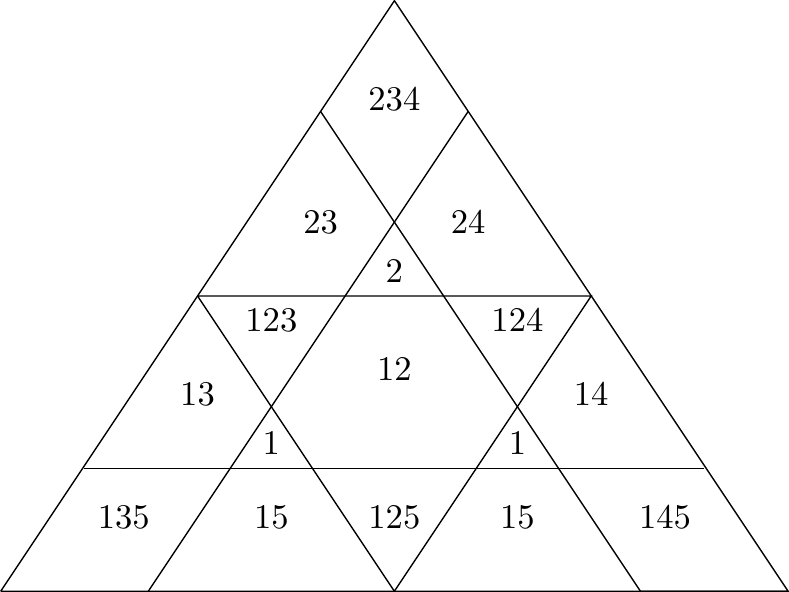}
		\caption*{$C20$}
	\end{subfigure}
	\begin{subfigure}{0.45\textwidth}
		\centering\includegraphics[width=.9\textwidth]{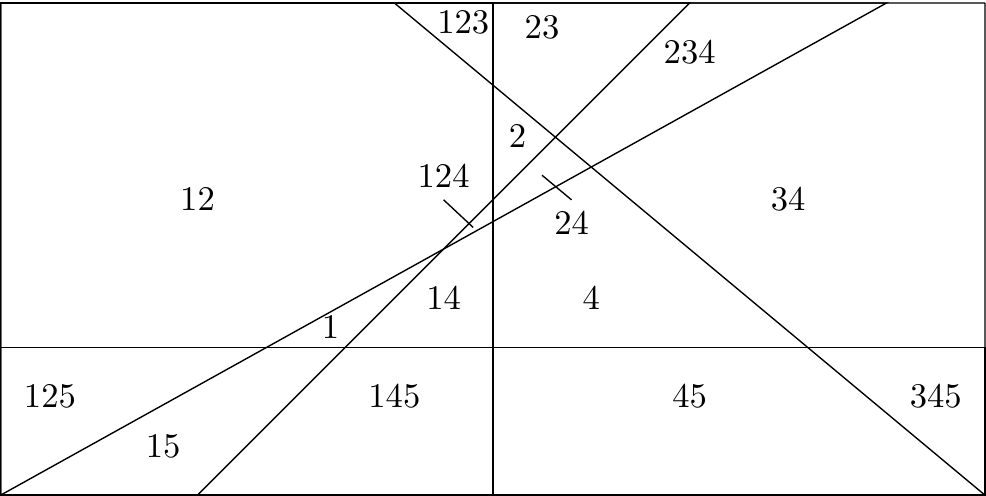}
		\caption*{$C21$}
	\end{subfigure}

\vspace{1cm}
	
	\begin{subfigure}{0.9\textwidth}
		\centering\includegraphics[width=.75\textwidth]{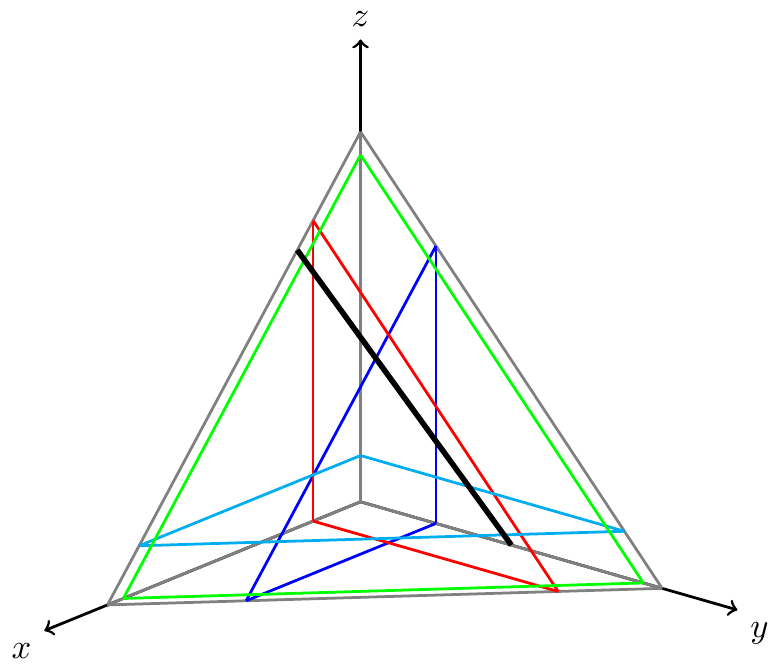}
		\caption*{$C22$}
	\end{subfigure}

\end{figure}

\end{document}